\documentclass[aps]{revtex4}
\usepackage{amsthm}
\usepackage{amssymb, amsmath}
\usepackage{color}
\usepackage{soul,xcolor}
\usepackage{graphicx}
\usepackage{url}
\usepackage{endnotes}
\usepackage[T1]{fontenc}

\begin{document}
\title{An exactly solvable quantum-metamaterial type model}
\author{A P Sowa$^{1}$, A M Zagoskin$^{2}$ }
\affiliation{$^{1}$Department of Mathematics and Statistics, University of Saskatchewan,
106 Wiggins Road,
Saskatoon, SK S7N 5E6,
Canada}
\affiliation{$^{2}${Department of Physics, Loughborough University, Loughborough, Leics LE11 3TU, UK}}

\maketitle
\newtheorem{definition}{Definition}
\newtheorem{theorem}{Theorem}
\newtheorem{proposition}{Proposition}
\newtheorem{lemma}{Lemma}
\newtheorem{corollary}{Corollary}
\newtheorem{algorithm}{Algorithm}
\newtheorem{conjecture}{Conjecture}

\begin{center}

 Abstract

 \end{center}
The key difficulty in the modelling of large quantum coherent structures lies in keeping track of nonlocal, multipoint quantum correlations between their constituent parts. Here we consider a special case of such a system, a fractal quantum metamaterial interacting with electromagnetic field, and show that it can be exactly solved by using a combination of the Haar transform and the Wigner-Weyl transform. Theoretical and experimental investigation of finite-size precursors  to exactly solvable fractal quantum structures  as this will help illuminate the behaviour of generic quantum coherent structures on a similar spatio-temporal scale.
\vspace{.2cm}

  \vspace{.2cm}

\vspace{.2cm}


\section{Introduction} \label{Introduction}

The dimensionality of the Hilbert space of a quantum system grows exponentially with the number of its degrees of freedom (e.g., quantum bits). The recognition of this fact, relatively long ago, led to the conclusion that classical means are inadequate for simulating large quantum systems and, as a logical consequence, ushered in conceptualization of quantum computers, \cite{Feynman, Manin1, Manin2}. On the other hand, the theoretical analysis of large quantum systems has been restricted to special, though important, cases (e.g., those of factorized or nearly factorized quantum states), which can be described using standard methods of quantum many-body theory.
At the same time, it is already possible to build at least partially quantum coherent arrays of quantum bits, \cite{Zagoskin2, Zagoskin3}. There are reasons to believe that the properties of such systems depend strongly on long-range quantum correlations, \cite{Rakhmanov, Zagoskin}. 
This highlights the need for tractable models of large-scale quantum arrays that would capture nonlocal quantum coherence.
In this work we introduce a model of that kind, pertaining to a  quantum metamaterial interacting with electromagnetic field. The key concept underlying our model is a continuum limit of the interaction Hamiltonian which turns out to be a nonlocal operator in $L_2[0,1]$ that is endowed with scale-wise self-similarity. Our main result is a description of its properties via an application of multiresolution analysis. Subsequently, we obtain an explicit solution of the \emph{quantum metamaterial/electromagnetic field} system dynamics in some regimes.  The crucial finding is that propagation of the field is strongly affected by the quantum state of the metamaterial. Remarkably, our model incorporates nonlocal quantum coherence effects even in a low-dimensional approximation, thus effecting a type of \emph{compression} of the essential (nonlocal) features of the system dynamics.

\section{The model}

A linear oscillator (a single electromagnetic field mode) interacting with a two-level system (qubit) is described by the Hamiltonian (see e.g. \cite{Shore_Knight}, \cite{Gerry-Knight}, or  \cite{Fink}):

\[
\mathcal{H} = I\otimes \mathcal{H}_F + \mathcal{H}_q\otimes I +\mathcal{H}_I:\quad
\mathbb{H}_Q\otimes\mathbb{H}_F \longrightarrow
\mathbb{H}_Q\otimes\mathbb{H}_F.
\]
Here $\mathbb{H}_F = L_2(\mathbb{R})$ is the Hilbert space of the oscillator (e.g., a mode of electromagnetic field in a cavity), and  $\mathbb{H}_Q = \mbox{span}\{|g\rangle, |e\rangle\}$ that of the qubit. The field and qubit Hamiltonians have the standard form,
\begin{equation}
\mathcal{H}_F = \omega\left(\hat{a}^{\dagger}\hat{a} + \frac{1}{2}\right), \mathcal{H}_q = \frac{\Omega}{2}\left(|e\rangle\langle e|-|g\rangle\langle g|\right) \equiv -\frac{\Omega}{2}\sigma_z,
\label{eq:01}
\end{equation}
 where $\hat{a}$ ($\hat{a}^{\dagger}$) are bosonic creation (annihilation) operators, $[\hat{a},\hat{a}^{\dagger}]=1$, and $\sigma_z$ is one of the Pauli matrices. We can define the operators $\hat{a}$ in such a way that the electric field amplitude is proportional to $(\hat{a}+\hat{a}^{\dagger})$.

The interaction term describes the dipole-field interaction:
\begin{equation}\label{JC_HI}
\mathcal{H}_I = \lambda\, \sigma_x\otimes (\hat{a} + \hat{a}^\dagger):\quad \mathbb{H}_Q\otimes\mathbb{H}_F
\longrightarrow \mathbb{H}_Q\otimes\mathbb{H}_F;
\end{equation}
here
$\sigma_x = \sigma_+ + \sigma_-$, and $\sigma_+ = |e\rangle\langle g|$ and $\sigma_- = |g\rangle\langle e|$.

Generalizing to $K$ qubits, not interacting with each other, we get the interaction term in the form:
\begin{equation}\label{JC_HImult}
\mathcal{H}_K =  C_K\otimes (\hat{a} + \hat{a}^\dagger):\quad \left(\bigotimes_{k=1}^K\mathbb{H}_Q\right)\otimes\mathbb{H}_F
\longrightarrow \left(\bigotimes_{k=1}^K\mathbb{H}_Q\right)\otimes\mathbb{H}_F,
\end{equation}
where
\[
C_K = \sum_{k = 1}^{K}\lambda_k\, \sigma_x^{k}.
\]
Here, $\sigma_x^{k} = I\otimes \ldots I\otimes \sigma_x \otimes I \ldots$, where $I$ denotes the identity acting in $\mathbb{H}_Q$, and the Pauli matrix $\sigma_x$ is placed in the sequence in the $k^{\mbox{th}}$  place. Numerical experimentation shows that the matrix of $C_K$ looks like a discrete model of a certain self-similar set reminiscent of the famous Cantor set (Fig. 1).

\section{A continuum approximation for a self-similar chain of qubits} \label{subsection_continuum_lim}

We now embark upon examination of the continuum limit of a system with infinitely many qubits. In this limit, the space $\bigotimes_{k=1}^K\mathbb{H}_Q$ is replaced by $\mathbb{H}_{QMM} = L_2[0,1]$. In order to obtain a meaningful limit of the operator $C_K$, denoted $C: \mathbb{H}_{QMM}\rightarrow \mathbb{H}_{QMM}$, we assume that
\begin{equation}\label{scaling}
  \quad \lambda_k = \frac{\lambda}{2^k}
\end{equation}
For a finite-size approximation to our ideal system this scaling can be realised by the proper placement of qubits with respect to the (anti)nodes of the electromagnetic field mode. The corresponding operator $C$ is defined by first giving its explicit construction for a test function, and later extending it by continuity to the entire Hilbert space. To this end, let $\Phi = \Phi(x)$ be either continuous in $[0,1]$ or piecewise constant\footnote{The point is such functions are defined not just almost everywhere but everywhere.}. Moreover, consider $x\in (0,1]$ and its dyadic expansion $x = \sum_k \alpha_k(x)/2^k$. To avoid ambiguity we assume that dyadic rationals will always have an infinite expansion, e.g. $1/2 = .011111\ldots$. Then define
 \begin{equation}\label{def_C}
  C[\Phi](x) = \sum_k \frac{1}{2^k}\Phi\left(x  + \frac{(-1)^{\alpha_k(x)}}{2^k}\right)\quad x \in (0,1].
\end{equation}
Below, we will demonstrate that $C$ is bounded in the $L_2$ norm and, therefore, it can be extended to all of $\mathbb{H}_{QMM}$ by continuity. This operator realises the continuum limit of the interaction Hamiltonian (\ref{JC_HImult}); namely
\begin{equation}\label{JC_continuum}
\mathcal{H}_\infty =  \lambda\, \hat{C}\otimes (\hat{a} + \hat{a}^\dagger):\quad L_2[0,1]\otimes\mathbb{H}_F
\longrightarrow L_2[0,1]\otimes\mathbb{H}_F.
\end{equation}


\begin{figure}[ht!]
\includegraphics[width=160mm]{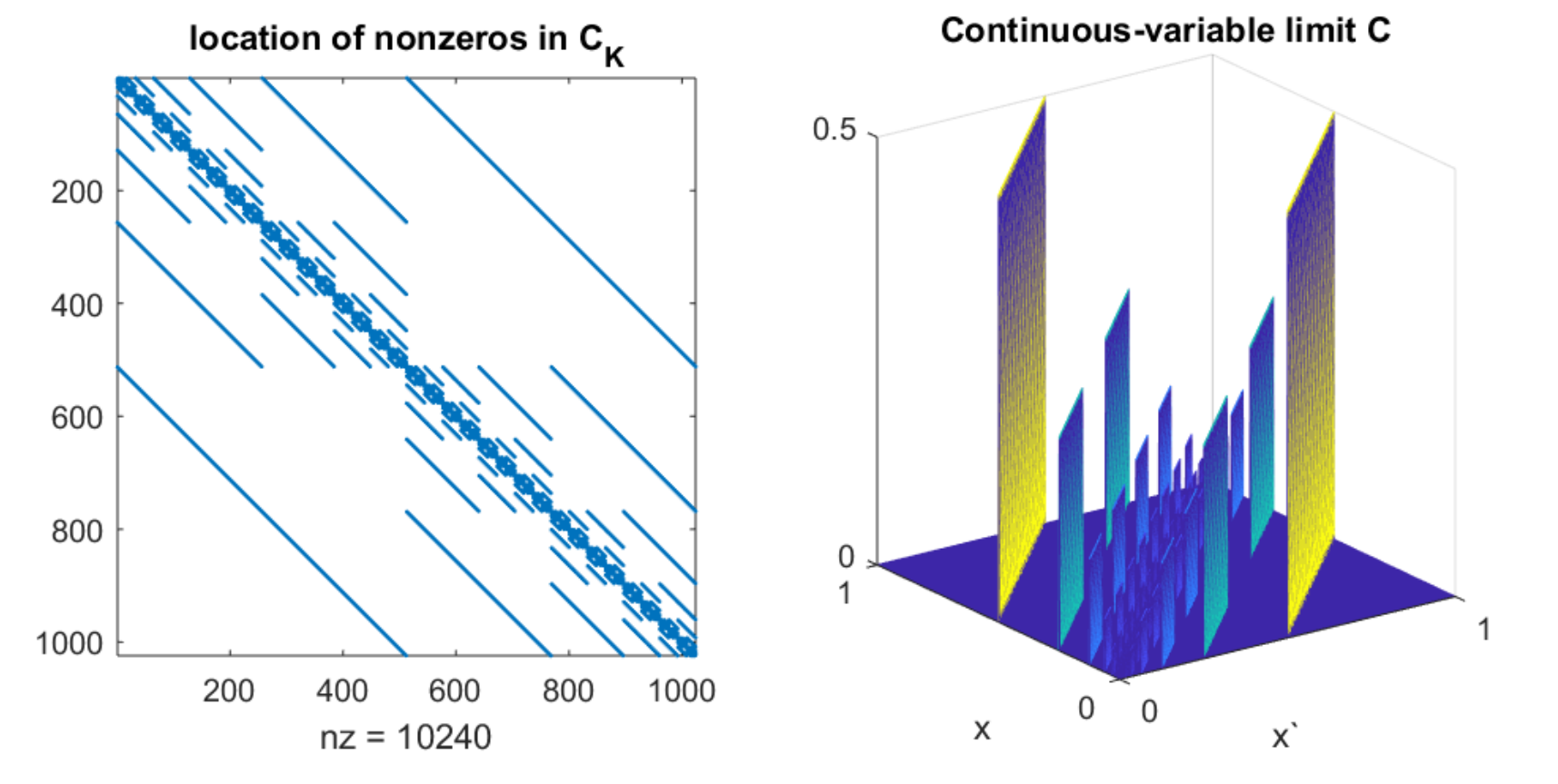}
\caption{The transition to the continuous limit with $C_K$: (a) The location of all nonzero entries of $C_K$ for $K=2^{10}$. Observe the sparsity (only $n_z = 10240$ entries are nonzero) as well as a suggestion of self-similarity. Indeed the diagonal quarter sub-blocks have the appearance of the entire matrix,  their diagonal sub-blocks, too, etc. (b) $C$ is obtained by going to the continuous limit. In the plot the vertical   bars are placed on the $(x,x')$ square $(0,1]^{\times 2}$. The bars may be interpreted as Dirac distributions, effecting the same outcome as the definition (\ref{def_C}) when $C$ acts on a continuous function. Note that, unlike for $C_K$, the self-similarity of $C$ is exact.}
\label{fig1}
\end{figure}

In a similar manner, we generalize the Hamiltonian $\mathcal{H}_q$ given in (\ref{eq:01}) to a $K$-qubit system Hamiltonian $V_K$, choosing  a  scaling of qubit excitation energies as that of coupling constants in (\ref{scaling}); namely
\begin{equation}\label{def_V_K}
  V_K = \sum_{k = 1}^{K}\frac{-\Omega_k}{2}\, \sigma_z^{k},  \quad\mbox{ where} \quad \Omega_k = \frac{\Omega}{2^k}.
\end{equation}
This specific scaling is chosen to ensure convergence, and in a finite-size system it can be realised through the qubit design. Note that $V_K$ is a diagonal matrix. In the continuum limit its action is replaced by the multiplicative potential, i.e.
  \begin{equation}\label{def_V}
  V[\Phi](x) = V(x) \Phi(x) = \left(x - \frac{1}{2}\right) \Phi(x).
\end{equation}
(The details of this observation are given in Theorem \ref{prop_V}.)
At this stage we can postulate the Hamiltonian for the entire system. First the underlying Hilbert space is $\mathbb{H}_{QMM} = L_2(0,1]\times \mathbb{H}_F = L_2(0,1]\times L_2(\mathbb{R})$; we will use the variable $x$ in the unit interval and variable $y$ in the line. The Hamiltonian acting in this space is given by:
  \begin{equation}\label{def_H_QMM}
\mathcal{H}_{QMM} = I\otimes \mathcal{H}_F + \lambda\, C \otimes (\hat{a} + \hat{a}^\dagger) +  V\otimes I.
\end{equation}
Observe that $\mathcal{H}_F$ is a differential operator in the $y$ variable, $V$ is linear potential in the $x$ variable, and $\hat{a} + \hat{a}^\dagger = \sqrt{2m\omega} \, \hat{y}$ is linear  in $y$. The interesting component is $C$, which is a nonlocal operator, acting via the $x$ variable, with an inherent self-similar structure.
\vspace{.5cm}

\noindent \emph{Remark 1.} Note that $C_K$ and $V_K$ are unitarily equivalent, up to scale. Indeed, since $u (-\sigma_z) u = \sigma_x$ for a unitary and self-adjoint matrix $ u = (\sigma_x - \sigma_z)/\sqrt{2}$, one has  $U*C_K*U \propto 2 V_K$ for the unitary and self-adjoint $U := u\otimes u\otimes \ldots \otimes u$. It is natural to ask if $C$ and $V$ also share some properties that unitarily equivalent ones would. In Section \ref{Section_C_Haar} we demonstrate that $C$ and $V$ have the same spectrum.
\vspace{.5cm}

It can be argued that, by extension of the finite-dimensional case, points in $[0,1]$ should be interpreted as factorized states of the qubit array. In this way, $x = \sum_k \alpha_k(x)/2^k$  corresponds to the factorized state $| \alpha_1 \,\alpha_2 \,\alpha_3 \,\ldots \rangle $ where we have identified $|0\rangle = |g\rangle$ and $|1\rangle = |e\rangle$. However, of course, no orthonormal basis of $L_2[0,1]$ exists whose elements could be identified with points in $[0,1]$. For one thing, a set of basis functions is countable while the interval is a continuum. This means that there is no natural way of representing \emph{factorized} states of an infinite array of qubits. Interpreting it more broadly, we might say that there is a zero probability that an infinite array of qubits will settle in a factorized state.

For a similar reason as above, the multiplier $V$ has no eigenfunctions in $L_2[0,1]$. On the other hand, it follows from Theorem \ref{theor_C}  (Section \ref{Section_C_Haar}) that operator $C$ does have  eigenfunctions (some examples are displayed in Fig. \ref{fig3}). Therefore, $C$ and $V$ are not unitarily equivalent. Moreover, the eigenfunctions of $C$ may only be regarded as nonlocal, i.e. not factorized.

\section{The quantum metamaterial Hamiltonian in the continuum limit} \label{Section_C_Haar}

First, recollect the structure of the Haar basis in $L_2[0,1]$, see e.g. \cite{Haar}, \cite{Wojtaszczyk}. Let $G(x) \equiv 1$ for $x\in(0,1]$ and $G(x)=0$ everywhere else on the real line. We will use notation $$G_{n,k}(x) = 2^{n/2} G(2^nx-k);$$ in particular $G_{0,0} = G$. Furthermore, let $H(x) = [G_{1,0}(x) -  G_{1,1}(x)]/\sqrt{2}$, and $$H_{n,k}(x) = 2^{n/2} H(2^nx-k).$$ Note that $H_{0,0} = H$. The following fundamental facts are  well known:

\begin{enumerate}
  \item
 Denote $V_n =\mbox{ span } \{G_{n,k}: k = 0, 1, 2,  \ldots 2^n-1 \}$. Then, $V_0 \subseteq V_1 \subseteq V_2 \ldots$, and
  \begin{equation*}
       L_2[0,1] = \bigcup_{n=0}^\infty V_n \quad (\mbox{multiresolution ladder}).
  \end{equation*}
  \item
  For $n = 0, 1, \ldots $ let $W_n = V_{n+1}/V_n$. Then $W_n = \mbox{ span } \{H_{n,k}: k = 0, 1, 2,  \ldots 2^n-1 \}$, so that
  \begin{equation}\label{direct_sum}
  L_2[0,1] = V_0\oplus \bigoplus_{n=0}^\infty W_n \quad (\mbox{direct sum decomposition}).
  \end{equation}
 In what follows we will make use of the orthogonal projections $\Pi_n: V_{n+1} \rightarrow W_n$.
    \item
It follows that the set of functions $\{G_{0,0}\}\cup \{H_{n,k}:  k = 0, 1, 2,  \ldots 2^n-1; n = 0,1,2, \ldots\}$ furnishes an orthonormal basis in $L_2[0,1]$ (the Haar basis). In all references to this basis we will assume the canonical order in $W_n$ to be according to increasing $k$, so that, overall, the order of basis functions is fixed to be:
\[
G_{0,0}, H_{0,0}, H_{1,0}, H_{1,1}, H_{2,0}, H_{2,1}, H_{2,2}, H_{2,3}, \ldots
\]
 We let $\mathcal{T_H}: L_2[0,1] \rightarrow \ell_2$ denote the Haar transform which assigns to a square integrable function, say, $f$ its ordered sequence of Haar coefficients:
 \[
c_0 =\int_{0}^{1} f(x)\, dx,  \mbox{ and }  c_{n,k} = \int_{0}^{1} f(x) H_{n,k}(x)\, dx.
\]
Clearly, $\mathcal{T_H}$ is a unitary transformation.
\end{enumerate}
The basic properties of $C$ are obtained by representing it in the Haar basis. Namely, we have the following:

\begin{theorem} \label{theor_C} Operator $C$ defined in (\ref{def_C}) extends to a continuous operator $C: L_2[0,1] \rightarrow L_2[0,1]$. More precisely, $C$ has the following properties:
\begin{enumerate}
\item
Let $D_n$ be the matrix of $\Pi_n C \Pi_n$ in the canonical basis $(H_{n,k})_{k=0}^{2^n-1}$. Then\footnote{We generally denote the $k\times k$ identity matrix by $I_k$, but skip the index when the matrix size is clear from context.},
  \begin{equation}\label{block_D}
  D_0 = [0], \quad \mbox{ while } \quad
    D_{n+1} = \frac{1}{2} \left[
                            \begin{array}{cc}
                              D_n & I \\
                              I & D_n \\
                            \end{array}
                          \right]\quad \mbox{ for } n\geq 0.
    \end{equation}
Also, if $n>0$, then $D_n$ is invertible, and its complete list of the eigenvalues is
\begin{equation}\label{eigsD}
\{\pm (2k+1)/2^n: \, k = 0, 1, \ldots 2^{n-1}-1\}.
\end{equation}
In particular, $C$ preserves the direct sum decomposition (\ref{direct_sum}), i.e.  $C V_0 = V_0$ and $C W_n \subseteq W_n$ for all $n$. Specifically,
\begin{equation}\label{blocks_C}
C  =\mathcal{T_H}^\dagger \left(I_1\oplus \bigoplus_{n=0}^\infty \, D_n \right)\mathcal{T_H}.
\end{equation}

\item
Let $E_n$, $n\geq 1$, be the diagonal matrix whose diagonal entries are as in (\ref{eigsD}) in the increasing order, and let $E_0 =[0]$. Also, let $u = (\sigma_x - \sigma_z)/\sqrt{2}$, and let $\mathcal{B}$ be the unitary operator in $\ell_2$ defined by $\mathcal{B} : = I_1\oplus I_1 \oplus \bigoplus_{n=1}^\infty \, u^{\otimes n}$.   Then,
\begin{equation}\label{diag_C}
   C = \mathcal{T_H}^\dagger\mathcal{B}^\dagger\left(I_1\oplus \bigoplus_{n=0}^\infty \, E_n\right)\mathcal{B}\,\mathcal{T_H}.
\end{equation}

\item
$C$ is a bounded self-adjoint operator in $L_2[0,1]$; its spectrum is $\sigma (C) = [-1/2, 1/2]$, and its norm $\|C\| = 1/2$.
\end{enumerate}
\end{theorem}

\begin{proof} \emph{1.}
  First, whenever a function, say, $\Phi$ is piecewise constant, definition (\ref{def_C}) yields a well defined $C[\Phi]$. In particular $C$ has well-defined value for every Haar basis function. Clearly, $C[G] = G$ and, hence, $V_0$ is invariant under $C$.
  Next, taking into account the self-similar structure of $C$, it is easy to observe that $CV_n \subseteq V_n$ for all $n$. Accordingly, let $C_n$ be the matrix of $C|_{V_n}: V_n \rightarrow V_n$ in the basis
   $\{G_{n,k}:  k = 0, 1, 2,  \ldots 2^n-1\}$ ordered according to increasing $k$. One verifies directly that
      \[
   C_{1}= \frac{1}{2}\left[
          \begin{array}{cc}
             1 & 1 \\
             1 & 1\\
          \end{array}
        \right].
   \]
   Moreover, self-similarity implies the recurrence:
   \[
   C_{n+1}= \frac{1}{2}\left[
          \begin{array}{cc}
             C_n & I \\
             I &  C_n\\
          \end{array}
        \right].
   \]
Next, let $D_n$ be the matrix of $\Pi_n C_{n+1}\Pi_n$ in the basis $\{H_{n,k}:  k = 0, 1, 2,  \ldots 2^n-1\}$ (ordered according to increasing $k$).  Consider the $2^{n+1}\times 2^n$ matrix:
   \[
    J_n = \left[
            \begin{array}{cccc}
              1 &  &  &  \\
              -1 &  &  &  \\
               & 1 &  &  \\
               & -1 &  &  \\
               &  & \vdots &  \\
               &  & \vdots &  \\
               &  &  & 1 \\
               &  &  & -1 \\
            \end{array}
          \right],
    \]
      whose $l^{\mbox{th}}$ column consists of the coordinates of $H_{n,l}$ in the basis $(G_{n+1,k})$ of $V_{n+1}$. Hence, $D_n = J_n^\dagger C_{n+1} J_n$, which implies (\ref{block_D}), including $D_0 = 0$.

     Note that $D_1 = -\sigma_x/2$ and its eigenvalues are $\pm 1/2$. In general, let $\mu_{n,k}: k = 0, 1, \ldots 2^{n-1} -1$ be the eigenvalues of $D_n$. It follows from (\ref{block_D}) that $\mu$ is an eigenvalue of $D_{n+1}$ if and only if $\mu = (\mu_{n,k} \pm 1)/2$. Statement (\ref{eigsD}) follows from this observation by induction. Statement (\ref{blocks_C}) summarizes these findings.
     \vspace{.2cm}

     \noindent
      \emph{ 2.} Rewrite recurrence (\ref{block_D}) in the form
\[
D_{n+1} = \frac{1}{2} I_2\otimes D_n + \frac{1}{2} \sigma_x \otimes I_{2^n}.
\]
Observe that $u = (\sigma_x - \sigma_z)/\sqrt{2}$ is unitary and self-adjoint and diagonalizes $D_1$, i.e. $u^\dagger D_1 u = u D_1 u =  - \sigma_z/2$. It follows by induction that the unitary transformation $u^{\otimes n}$ diagonalizes $D_n$, and that $E_n = u^{\otimes n}D_nu^{\otimes n}$ has the eigenvalues of $D_n$ on its diagonal in the increasing order from top-left to bottom-right. This together with (\ref{blocks_C}) implies (\ref{diag_C}).
\vspace{.2cm}

     \noindent
    \emph{ 3.} The representation of $C$ given by (\ref{diag_C}) shows that it is a bounded and self-adjoint operator in $L_2[0,1]$ whose spectrum $\sigma (C)$ is the closure of the set of eigenvalues, which is the interval $[-1/2, 1/2]$. It is well known that for a self-adjoint operator its spectral radius is equal to its norm, hence $\|C\|=1/2$. This completes the proof of the theorem.
\end{proof}
\vspace{.2cm}

\noindent \emph{Remark 1.} One of the benefits of diagonalization (\ref{diag_C}) is a direct expression for $\exp itC$.
\vspace{.2cm}

\noindent \emph{Remark 2.} Definition (\ref{def_C}) and Theorem \ref{theor_C} indicate two alternative ways of implementing operator $C$ numerically. It ought to be emphasized that  a finite-dimensional approximation is bound to introduce some artifacts. In particular, the one based on a finite Haar transform  will yield a small (and decreasing with the dimension) diagonal component of $C$. (Note that the infinite dimensional $C$ as defined in (\ref{def_C}) has no ``diagonal" term.) A representation of $C$ in the Haar basis is illustrated numerically in Fig. 2.
\vspace{.5cm}

Next, we justify formula (\ref{def_V}). To this end it is convenient to introduce the periodized Haar function $H^\#$, defined by
\[
H^\#(x) = H(x\, \mbox{mod } 1)\quad \mbox{ for all } x\in \mathbb{R}.
\]
Note a relation to the Rademacher functions $r_{n} = \mbox{sgn} \,\sin (2^{n}\pi x)$, namely
\[
H^\#(2^{n}x) = r_{n+1}(x), \quad n= 0,1,\ldots
\]
 The continuum limit of $V_K$, as defined in (\ref{def_V_K}), is obtained via the following steps:
\begin{enumerate}
  \item Replace  $\sigma_z^{k} $ by $H^\# (2^{k-1} x)$. In particular, this has the effect of reducing the diagonal matrix $V_K$ to a function defined in $(0,1]$ or, indeed, almost everywhere (\emph{a.e.}) in $\mathbb{R}$.
  \item Pass to the limit as $K\rightarrow \infty$.
Thus, the continuum limit of $V_K$, denoted $V(x)$, is defined by
\begin{equation}\label{def_V_rigor}
 V(x) =  -\Omega \sum_{k=1}^\infty \frac{1}{2^{k+1}}H^\# (2^{k-1} x) .
\end{equation}
Thus, $V(x)$ is a periodic function (with period $1$) defined \emph{a.e.} in $\mathbb{R}$. However, in the applications, we only consider its restriction to the unit interval, $x\in (0,1]$.
\end{enumerate}

We are now in a position to demonstrate:

\begin{theorem}\label{prop_V}
The continuum limit of the free qubit Hamiltonian defined by (\ref{def_V_rigor}) satisfies
\[
V(x) =  \Omega\, \left(x-\frac{1}{2}\right) \,\mbox{ a.e. } \mbox{ in } (0,1].
\]
\end{theorem}

\noindent \emph{Proof.} First, observe the identity
\[
H^\# (2^{k} x) = \frac{1}{2^{k/2}}\sum_{l=0}^{2^{k}-1} H_{k,l}(x) \quad \mbox{ a.e. in }  (0,1].
\]
(Note that the right hand side contains the sum of all the Haar functions that furnish the basis of $W_k$.) This together with (\ref{def_V_rigor}) indicates that the Haar coefficients of $V(x)$ are given by
\[
c_0 = 0, \mbox{ and }  \, c_{k,l} = - 2^{-3k/2 - 2}\, \mbox{ for all }\, k \geq 0.
\]
On the other hand, a direct calculation gives
\[
\int_{0}^{1} \left(x-\frac{1}{2}\right) H_{k,l}(x) \, dx = - 2^{-3k/2 - 2}.
\]
Finally, since the Haar coefficients of $V(x)$ coincide with those of $x-1/2$, the two functions are equal (\emph{a.e.}). $\Box$
\vspace{.5cm}

\begin{figure}[ht!]
\includegraphics[width=90mm]{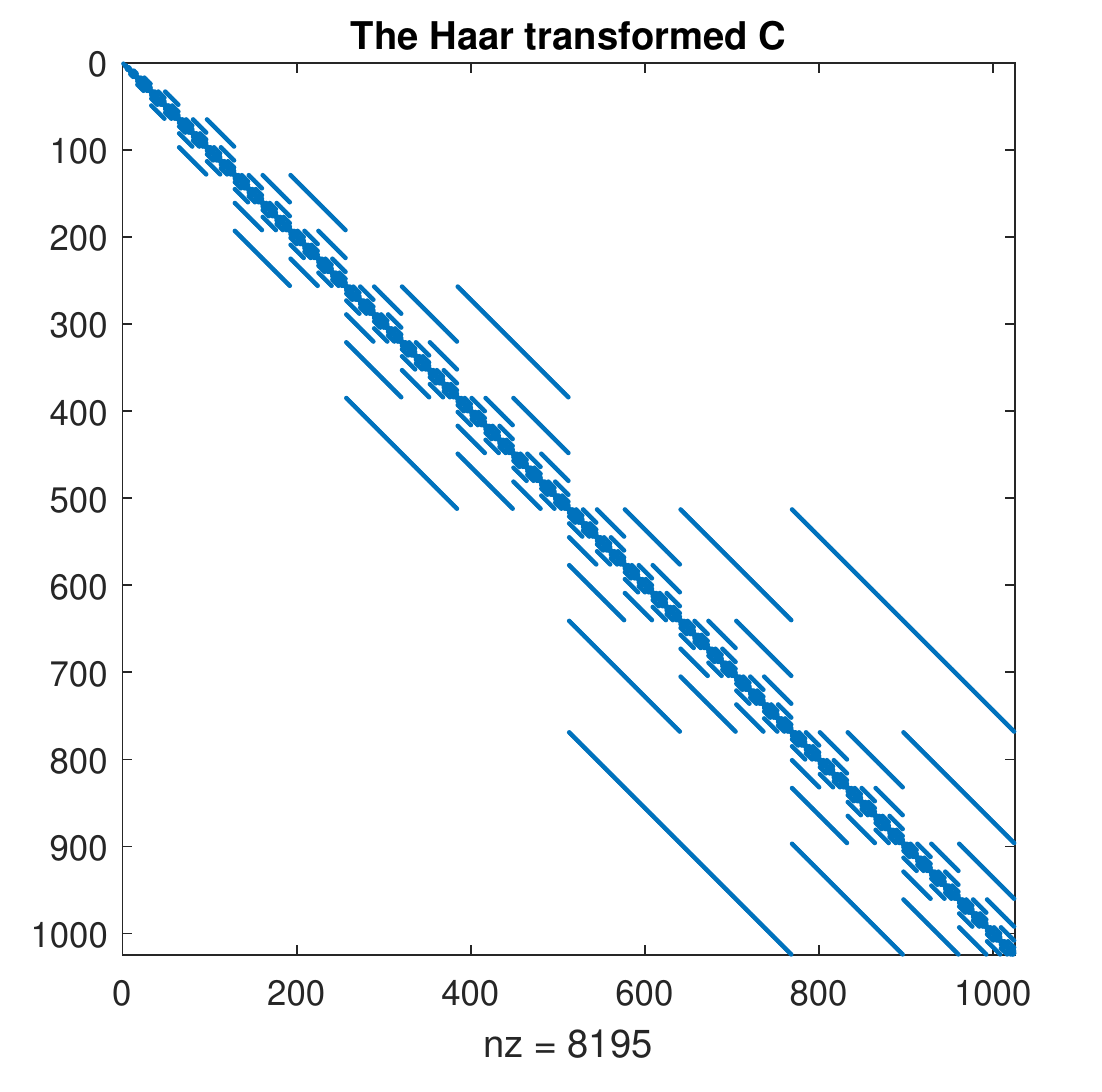}
\caption{Nonzero elements of $\mathcal{T_H} C \mathcal{T_H}^\dagger =  I_1\oplus \bigoplus_{n=0}^\infty \, D_n $ (i.e. $C$ represented in the Haar basis). Observe the ``telescope" structure. Also, the displayed finite-dimensional block ($K=10$ qubits) has relatively fewer nonzero elements than $C_K$ (indeed,  $nz = 8195$), i.e. the Haar basis effects a numerical compression of $C_K$. Most importantly, the Haar representation makes it possible to demonstrate that the eigenvectors of $C$ are spread over all qubits (delocalized), cf. Fig 3. This in turn sheds light on the role of long-range coherence in the dynamic of a quantum metamaterial coupled to a resonator, see Section \ref{Section_explicit_dynamics}.}
\label{fig2}
\end{figure}

\section{Explicit solutions of the system dynamic in the limit $\Omega \to 0$} \label{Section_explicit_dynamics}

Some crucial properties of the Hamiltonian (\ref{def_H_QMM}) become manifest in the case when $\Omega$ (scale of qubit excitation energies) is  negligible compared to the energy of the field mode, $\omega$, and the qubit-field interaction scale, $\lambda$.
In the light of Theorem \ref{prop_V} this amounts to modifying the Hamiltonian by erasing a compact component. Thus, setting $\Omega = 0$ we examine
  \begin{equation}\label{def_H_QMMprime}
\mathcal{H'}_{QMM} = I\otimes \mathcal{H}_F + \lambda\, C \otimes (\hat{a} + \hat{a}^\dagger).
\end{equation}
Let $\rho:  \mathbb{H}_{QMM}\otimes \mathbb{H}_{F} \rightarrow \mathbb{H}_{QMM}\otimes \mathbb{H}_{F}$ represent the state of the QMM and field system. It evolves according to the Heisenberg equation:
\[
 i\,\partial_t \rho = [ \mathcal{H'}_{QMM} , \rho].
\]
We make an Ansatz $$\rho = |\Phi_{n,k,s}\rangle\langle \Phi_{n,k,s}| \otimes \rho_F.$$ Here, $\rho_F: \mathbb{H}_F \rightarrow\mathbb{H}_F$; $s = \pm$ and $\Phi_{n,k,s}$ is an eigenstate of $C$ corresponding to the eigenvalue $E_{n,k,s} = s (2k+1)/2^n$, where $k = 0, 1,2 \ldots 2^{n-1}-1$ for $n = 1,2,\ldots$. Thus, we look for such solutions of the Heisenberg equation in which the state of the QMM is determined. Substituting into (\ref{def_H_QMMprime}) we find that the field part $\rho_F$ satisfies
\[
 i\,\partial_t \rho_F = [  \mathcal{H}_F + \lambda E_{n,k,s} (\hat{a} + \hat{a}^\dagger), \rho_F].
\]
In order to analyze solutions of this equation it is convenient to pass to the Wigner representation of $\rho_F$, see e.g. \cite{Gardiner-Zoller} or \cite{Gosson}. Namely, set
\[
f(q,p) = \int_\mathbb{R}d\xi_1\int_\mathbb{R}d\xi_2\,
e^{-2\pi i (\xi_1q + \xi_2p)}\, \mbox{ Tr }\left(W_{\xi_1,\xi_2}\, \rho_F\right),
\]
where $ W_{\xi_1,\xi_2} = e^{ 2\pi i (\xi_1\hat{q} + \xi_2 \hat{p})} $. As is well-known the transform is reversible, so that $f$ and $\rho_F$ are in one-to-one correspondence. A calculation shows that the Heisenberg equation is transformed into the following first-order PDE:
\[
\partial_t f = (q + \lambda E_{n,k,s})\, \partial_p f - p\, \partial_q f.
\]
This is easily solved via the method of characteristics. It turns out that the characteristics are circles centered at $(-\lambda E_{n,k,s},0)$ in the $(q,p)$ plane. In other words,
$
f(t, q, p) = f(0, q(-t), p(-t))$, where
\[
\left[
  \begin{array}{c}
    q(t) \\
    p(t) \\
  \end{array}
\right] =
\left[
  \begin{array}{cc}
    \cos t  & -\sin t \\
    \sin t & \cos t \\
  \end{array}
\right]
\left[
  \begin{array}{c}
    q +\lambda E_{n,k,s} \\
    p \\
  \end{array}
\right]
-
\left[
  \begin{array}{c}
    \lambda E_{n,k,s} \\
    0 \\
  \end{array}
\right].
\]

  This reveals how light propagation responds to the QMM state. Namely, by selecting an appropriate QMM state, say, by measurement on this subsystem, one automatically selects the center of rotation in the Wigner-Weyl  $(q,p)$ plane. These points form a dense subset of the interval $[-\lambda, \lambda ]$ in the $q$-axis.  Examples of eigenstates $|\Phi_{n,k,s}\rangle$ are given in Fig.\ref{fig3}; these are Walsh-type functions.
\vspace{.2cm}

\noindent
\emph{Remark.}
The effect of QMM state on field propagation may be described more abstractly via a shift of the Hamiltonian. Indeed, one may observe
$$
\mathcal{H}_F +  E_{n,k,s}\lambda\, (\hat{a} + \hat{a}^\dagger) = \mathcal{T}_{n,k,s}^\dagger \mathcal{H}_F \mathcal{T}_{n,k,s} - \frac{1}{2} (E_{n,k,s}\lambda )^2,
$$
where $ \mathcal{T}_{n,k,s}$ is the displacement operator $\exp (\lambda E_{n,k,s}\,\frac{d\,}{dQ})$ with $Q = \hat{a} + \hat{a}^\dagger$. Thus, changing the state of the metamaterial has the effect of shifting the field Hamiltonian in the position-momentum space along the position direction (to the left or to the right depending on the sign $s$).

\begin{figure}[ht!]
\includegraphics[width=160mm]{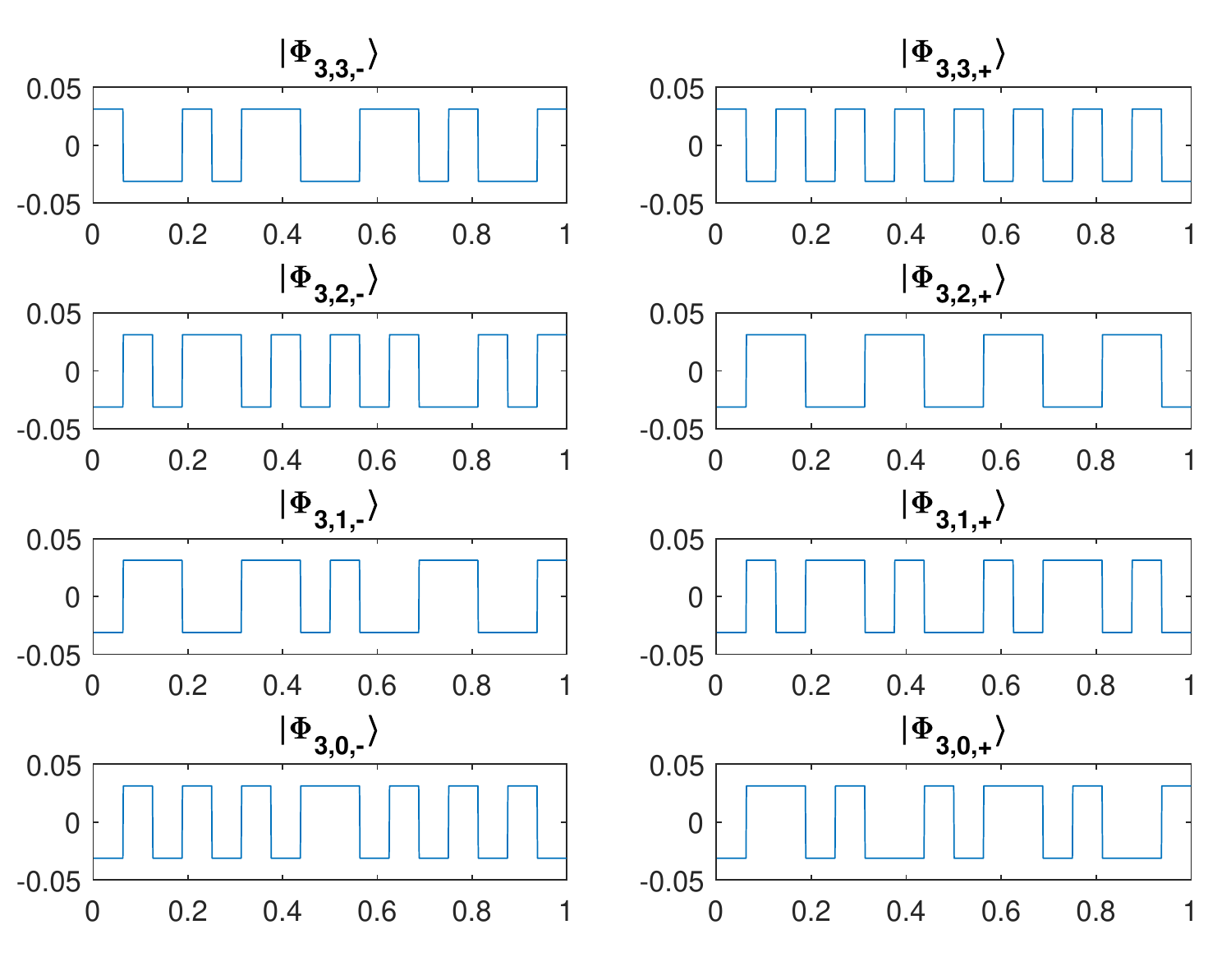}
\caption{Examples of eigenstates of the QMM in the $\Omega \sim 0$ regime, i.e. eigenfunctions of $C:L_2[0,1] \rightarrow L_2[0,1]$. }
\label{fig3}
\end{figure}

\section*{Summary}
We have examined  a continuum limit  model of a structure comprising an infinite array of qubits coupled to a single electromagnetic field mode in a certain specific way. Applying the methods of multiresolution analysis, we have found an explicitly solvable solution of its dynamic. The model incorporates nonlocal quantum coherence and, as such, provides fresh insights for further theoretical and experimental investigations of quantum coherent structures.

It is interesting to reflect on the evolving role of Harmonic Analysis in the study of periodic quantum structures. An application of the Fourier methods to quantum spin systems, originating in \cite{Baruch_1, Baruch_2}, is now routine. Another example is the Repeat Space Theory of large molecules and carbon nanotubes, \cite{ Arimoto2,  Arimoto4}, which draws on special properties of Toeplitz matrices, \cite{Taylor}.  At the same time, we know of no prior instance of a ``quantum" application of the Haar transform, which here emerged in the study of a scaled qubit array.

\section*{Acknowledgments}  We acknowledge helpful discussions  with Dr. Patrick Navez. We also acknowledge partial support of the University of Saskatchewan AMBASSADOR program. AZ was partially supported by the Ministry of Science and Higher Education of the Russian Federation in the framework of Increase Competitiveness Program of NUST \guillemotleft MISiS\guillemotright,  No. K2-2017-085.


\theendnotes

\end{document}